\documentclass[10pt,aps,twocolumn,longbibliography,superscriptaddress,prx]{revtex4-2}
\usepackage[english]{babel}
\usepackage[utf8]{inputenc}
\usepackage[colorinlistoftodos, color=green!40, prependcaption]{todonotes}
\usepackage{amsthm}
\usepackage{amssymb}
\usepackage{comment}
\usepackage{mathtools}
\usepackage{physics}
\usepackage{xcolor}
\usepackage{graphicx}
\usepackage[export]{adjustbox}
\usepackage{graphbox}
\usepackage{placeins}
\usepackage[T1]{fontenc}
\usepackage{lipsum}
\usepackage{centernot}
\usepackage{csquotes}
\usepackage[pdftex, pdftitle={Article}, pdfauthor={Author},colorlinks = true, linkcolor = blue, urlcolor  = blue, citecolor = red]{hyperref} 

\makeatletter
\@addtoreset{equation}{myequation}
\makeatother

\makeatletter
\@addtoreset{theorem}{mytheorem}
\makeatother

\makeatletter
\@addtoreset{section}{mysection}
\makeatother

\usepackage{makecell}
\usepackage{amsfonts}
\usepackage{amsmath}
\usepackage{stmaryrd}
\usepackage{braket}
\usepackage{dsfont}
\usepackage{bbold}
\usepackage{relsize}
\usepackage{float}
\usepackage[toc,page,titletoc]{appendix}
\usepackage[font=scriptsize]{caption}
\let\ACMmaketitle=\maketitle
\renewcommand{\maketitle}{\begingroup\let\footnote=\thanks \ACMmaketitle\endgroup}

\newcommand{\M}[1]{\mathcal{M}_{#1}}

\newcommand{\C}[1]{\mathbb{C}^{#1}}

\newcommand{\env}{\operatorname{env}}
\newcommand{\out}{\operatorname{out}}

\newtheorem{theorem}{Theorem}[section]

\newtheorem*{definition*}{Definition}

\newtheorem{corollary}[theorem]{Corollary}
\newtheorem{lemma}[theorem]{Lemma}
\newtheorem{remark}[theorem]{Remark}

\newtheorem*{conjecture*}{Conjecture}

\theoremstyle{definition}
\newtheorem{example}[theorem]{Example}

\newcommand{\iden}{\mathbb{1}}

\newcommand{\PRLsep}{\noindent\makebox[\linewidth]{\resizebox{0.3333\linewidth}{1pt}{$\blacklozenge$}}\bigskip}

\begin{document}
\title{Fully undistillable quantum states are separable}
\author{Satvik Singh}
\email{satviksingh2@gmail.com}
\affiliation{Department of Applied Mathematics and Theoretical Physics, \\ University of Cambridge, Cambridge, United Kingdom}

\author{Nilanjana Datta}
\email{n.datta@damtp.cam.ac.uk}
\affiliation{Department of Applied Mathematics and Theoretical Physics, \\ University of Cambridge, Cambridge, United Kingdom}

\date{\today}

\begin{abstract}

Assume that Alice, Bob, and Charlie share a tripartite pure state $\ket{\psi_{ABC}}$. We prove that if Alice cannot distill entanglement with either Bob or Charlie using $\ket{\psi_{ABC}}$ and local operations with any one of the following configurations for classical communication: $(A\to B, A\leftrightarrow C), (A\leftrightarrow B, A\to C),$ and $(A\leftrightarrow B, A\leftrightarrow C)$, then the same is also true for the other two configurations. Moreover, this happens precisely when the state is such that both its reductions on systems $AB$ and $AC$ are separable, which is further equivalent to the reductions being PPT. This, in particular, implies that any NPT bipartite state is such that either the state itself or its complement is 2-way distillable. To prove these results, we first obtain an explicit lower bound on the 2-way distillable entanglement of low rank bipartite states. Furthermore, we show that even though not all low rank states are 1-way distillable, a randomly sampled low rank state will almost surely be 1-way distillable.
\end{abstract}

\maketitle

\section{Introduction}

Entanglement is a resource of fundamental importance in quantum information theory. Use of entanglement makes numerous information-theoretic protocols possible in the quantum realm, which are classically impossible. These include quantum teleportation~\cite{Bennett1993teleportation}, various quantum cryptographic protocols~\cite{Ekert1991crypto, Gisin2002crypto}, and superdense coding~ \cite{Bennett1992superdense}. 
Since noise and decoherence are inescapable features of our world, purification (or distillation) of noisy entanglement is of key importance in the execution of many of the aforementioned protocols.

In an entanglement distillation setup \cite{Bennett1996purification, Bennett1996distillationQECC, Horodecki2001review}, Alice and Bob initially share $n$ copies of a mixed bipartite state $\rho_{AB}$. Their aim is to obtain, via a completely positive trace-preserving map $\Lambda_n$ comprising of local operations and 1-way $(A\to B)$ or 2-way classical communication $(A\leftrightarrow B)$ (1-LOCC or 2-LOCC), a state $\Lambda_n (\rho_{AB}^{\otimes n})$ that is close to $m_n$ copies of the Bell state $\Omega_2^+:=\ketbra{\Omega_2^+}$, where $\ket{\Omega_2^+}=(\ket{00} + \ket{11})/\sqrt{2}$. If $\norm{\Lambda_n (\rho_{AB}^{\otimes n}) - \Omega_2^{+\otimes m_n}}\to 0$ as $n \to \infty$, then the asymptotic rate $\lim_{n\to \infty} m_n/n$ of ebit generation is called an achievable rate for 1-way (resp.~2-way) entanglement
distillation. 
The 1-way (resp.~2-way) distillable entanglement $D_\rightarrow(\rho_{AB})$ (resp.~$D_\leftrightarrow(\rho_{AB})$) is defined as the supremum over all
achievable rates under 1-LOCC (resp.~2-LOCC). 
The inequality $D_{\leftrightarrow}(\rho_{AB})\geq D_{\to}(\rho_{AB})$ holds trivially since every 1-LOCC operation is also a 2-LOCC operation. The inequality can also be strict \cite{Bennett1996distillationQECC}.

The `environment' in the above scheme is modelled via a third party (say, Charlie), so that the three parties together share a purification $\ket{\psi_{ABC}}$ of $\rho_{AB}$. Alice then shares the \emph{complementary} state $\sigma_{AC} :=\operatorname{Tr}_B \ketbra{\psi_{ABC}}$ (also denoted as $\rho^c_{AC}$) with Charlie \footnote{The isometric freedom in choosing a purification $\ket{\psi_{ABC}}$ of $\rho_{AB}$ imparts a similar freedom in the choice of a complementary state $\sigma_{AC}$. However, it is easy to check that different choices of $\sigma_{AC}$ yield the same value of distillable entanglement $D_{\to}(\sigma_{AC})$ or $D_{\leftrightarrow}(\sigma_{AC})$.}. In the distillation of $\rho_{AB}$, as maximal entanglement is established between Alice and Bob, Charlie gets uncorrelated from both of them due to the monogamy of entanglement \cite{Coffman2000monogamy, Osborne2006monogamy}. However, if Alice and Charlie are also allowed to communicate classically, and if the end goal is for Alice to distill ebits at the best rate possible -- regardless of whether it is done with Bob or Charlie -- then Alice can either choose $\rho_{AB}$ or $\sigma_{AC}$ in the distillation scheme, depending on which state yields the best rate. This leads to the definition of the following \emph{maximal} distillation rates, one for each configuration of communication links between the three parties (Fig~\ref{fig:my_label}):
\begin{align}\label{eq:maxrates}
    D^{\leftrightarrow}_ {\leftrightarrow}(\psi_{ABC}) &:= \max \{D_{\leftrightarrow}(\sigma_{AC}), D_{\leftrightarrow}(\rho_{AB}) \}, \nonumber \\ 
    D^{\to}_{\leftrightarrow}(\psi_{ABC}) &:= \max \{D_{\to}(\sigma_{AC}), D_{\leftrightarrow}(\rho_{AB}) \}, \nonumber \\
    D^{\leftrightarrow}_{\to}(\psi_{ABC}) &:= \max \{D_{\leftrightarrow}(\sigma_{AC}), D_{\to}(\rho_{AB})\}, \nonumber \\
    D^{\to}_{\to}(\psi_{ABC}) &:= \max \{D_{\to}(\sigma_{AC}), D_{\to}(\rho_{AB}) \}.
\end{align}
Note that the above quantities do not depend on the particular choice of the purification $\ket{\psi_{ABC}}$. 
The state $\rho_{AB}$ (or equivalently the pure state $\psi_{ABC}= \ketbra{\psi_{ABC}}$) is said to be {\em{2-way fully undistillable}} if $D^{\leftrightarrow}_ {\leftrightarrow}(\psi_{ABC})=0$. Similarly, it is said to be (1,2)-way (resp.~(2,1)-way) fully undistillable if $D^{\to}_{\leftrightarrow}(\psi_{ABC})$ (resp.~$D_{\to}^{\leftrightarrow}(\psi_{ABC})$) is zero. 

It is worth emphasizing that in the above setting, once Alice chooses the party with whom she wishes to distill entanglement (say Bob), Alice and Bob can only communicate with each other and not with Charlie. Here is another way to look at this. We can think of Alice as a double agent who shares a tripartite state with Bob and Charlie, each of whom thinks that Alice is going to execute the distillation protocol with them. However, Alice is greedy and only wants to maximise her ebit output (irrespective of who it is shared with). She chooses either the reduced state with Bob or the one with Charlie, depending on which gives her the better distillation rate. Upon doing so, she cuts off all classical communication with the third party. This is what makes our scenario different from the environment-assisted distillation setups (see \cite{Smolin2005assistance} and references therein), where Charlie acts as a helper and is allowed to perform measurements and classically communicate with Alice and Bob to boost the rate at which they can distill entanglement among themselves.

It turns out that distilling entanglement is intimately connected with the task of reliably transmitting quantum information between two parties. If the quantum communication link from Alice to Bob is modelled by a quantum channel $\Phi:\M{d_A}\to \M{d_B}$ (here, $\M{d}$ denotes the set of all $d\times d$ complex matrices), the 1-way (resp. 2-way) \emph{quantum capacity} $\mathcal{Q}_{\to}(\Phi)$ (resp. $\mathcal{Q}_{\leftrightarrow}(\Phi)$) of $\Phi$ is the maximum rate at which Alice can reliably send quantum information to Bob by using asymptotically many copies of $\Phi$ and 1-way $(A\to B)$ (resp. 2-way) classical communication. If no classical communication between Alice and Bob is allowed, the corresponding quantity $\mathcal{Q}(\Phi)$ is simply called the \emph{quantum capacity} of $\Phi$. The bounds $\mathcal{Q}_{\leftrightarrow}(\Phi)\geq \mathcal{Q}_{\to}(\Phi)\geq \mathcal{Q}(\Phi)$ hold trivially. Surprisingly, 1-way communication does not help in transmitting quantum information, i.e., $\mathcal{Q}(\Phi)=\mathcal{Q}_{\to}(\Phi)$ \cite{Bennett1996distillationQECC}. However, 2-way communication can increase the quantum capacity of some channels, i.e., $\mathcal{Q}_{\leftrightarrow}(\Phi)>\mathcal{Q}_{\to}(\Phi)$ for these channels~\cite{Bennett1996distillationQECC}. In order to establish the link between quantum information transmission and entanglement distillation, we first define the \emph{Choi state}
\begin{equation}
J_{AB}(\Phi) := (\operatorname{id} \otimes \Phi)(\Omega^+_{d_A}),
\end{equation}
where $\ket{\Omega^+_{d_A}}:=(1/\sqrt{d_A})\sum_i\ket{ii}\in \C{d_A}\otimes \C{d_A}$ is maximally entangled and $\Omega_{d_A}^+ := \ketbra{\Omega^+_{d_A}}$.
The idea \cite{Bennett1996distillationQECC} is that Alice can locally prepare multiple copies of $\Omega^+_{d_A}$ and send one part of each copy to Bob via the channel $\Phi$, so that at the end they share multiple copies of the state $J_{AB}(\Phi)$ among themselves. Then, if Alice and Bob can distill ebits from $J_{AB}(\Phi)$ at an asymptotic rate $R$, they can use these ebits in the standard teleportation protocol to send quantum information from $A\to B$ at the same rate. This gives us the bounds:
\begin{align}\label{eq:cap-distill-bound1}
    \mathcal{Q}_{\leftrightarrow}(\Phi) \geq D_{\leftrightarrow} (J_{AB}(\Phi))\,;\, 
    \mathcal{Q}_{\to}(\Phi) \geq D_{\to} (J_{AB}(\Phi)).
\end{align}
On the other hand, if Alice and Bob initially share multiple copies of the Choi state
$J_{AB}(\Phi)$, they can implement the channel $\Phi$ with probability $1/d^2_A$ by using $J_{AB}(\Phi)$ in the standard teleportation protocol \cite{Bennett1996distillationQECC}, \cite[Section 2.1]{ Wolf2012Qtour}.
This fact can be used to obtain the following bound:
\begin{align}\label{eq:cap-distill-bound2}
\frac{1}{d^2_A} \mathcal{Q}^{(1)}(\Phi) \leq D_{\to} (J_{AB}(\Phi)),
\end{align}
where $\mathcal{Q}^{(1)}(\Phi)$ is the \emph{coherent information} of $\Phi$ and its regularization yields the 1-way quantum capacity (see Appendix~\ref{appen:capbound} for more details).

Let us now briefly describe the primary contribution of our work. Recall that a state $\rho_{AB}$ is called
\begin{enumerate}
    \item \emph{separable} if it lies in the convex hull of product states, i.e., states of the form $\sigma_A \otimes \gamma_B$;
    \item PPT if it has positive partial transpose;
    \item \emph{2-way undistillable} if $D_{\leftrightarrow }(\rho_{AB})=0$;
    \item \emph{1-way undistillable} if $D_{\to}(\rho_{AB})=0$.
\end{enumerate} 

\medskip

The implications $1\implies 2\implies 3\implies 4$ are well-known. However, $4\centernot\implies 3$ \cite{Bennett1996distillationQECC} and $2\centernot\implies 1$ \cite{Horodecki1996sep}. Deciding whether or not every non-PPT (NPT) state is 2-way distillable consitutes the fundamental NPT bound entanglement problem \cite{Shor2000evidence, Pankowski2010nptbound}. In this paper, we resolve an interesting variant of this problem. We prove that a state is \emph{2-way fully undistillable}, i.e., $D^{\leftrightarrow}_{\leftrightarrow}(\psi_{ABC})=0$, if and only if both $\rho_{AB}$ and its complement $\sigma_{AC}$ are separable, which is further equivalent to both of them being PPT. Clearly, 2-way full undistillability implies (2,1)- and (1,2)-way full undistillability:
\begin{align}
    \hspace{-10pt}D^{\leftrightarrow}_{\leftrightarrow}(\psi_{ABC})=0 \implies 0 &= D^{\leftrightarrow}_{\to}(\psi_{ABC}) = D^{\to}_{\leftrightarrow}(\psi_{ABC}).
\end{align}
However, it might be the case that $D^{\leftrightarrow}_{\to}(\psi_{ABC})=0$, but allowing backward classical communication from Bob to Alice makes distillation possible so that $D^{\leftrightarrow}_{\leftrightarrow}(\psi_{ABC})>0$. Surprisingly, we show that this is impossible
because of the following equivalences (Corollary~\ref{coro:main1}): 
\begin{align}
    D^{\leftrightarrow}_{\leftrightarrow}(\psi_{ABC})=0 &\iff D^{\leftrightarrow}_{\to}(\psi_{ABC})=0 \nonumber \\ &\iff D^{\to}_{\leftrightarrow}(\psi_{ABC})=0,
\end{align}
This uselessness of backward classical communication in the stated entanglement distillation setup is quite counter-intuitive and is the crux of our paper. We also prove that the above equivalences break down for 1-way full undistillability, i.e., there exist states ${\psi_{ABC}}$ such that $D^{\to}_{\to}(\psi_{ABC})=0$ but $D^{\leftrightarrow}_{\leftrightarrow}(\psi_{ABC})>0$ (Corollary~\ref{coro:main3}). 

\onecolumngrid
\medskip

\begin{figure}[H] \label{fig:I}
    \centering
    \includegraphics[scale=1.2]{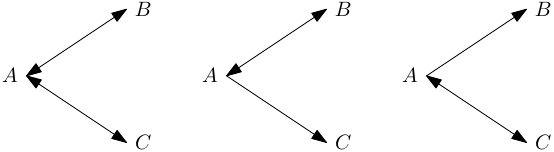}
    \caption{Alice, Bob, and Charlie share a tripartite pure state $\ket{\psi_{ABC}}$ and can classically communicate in any one of the three configurations shown above. The task is for Alice to distill maximally entangled states with either Bob or Charlie. Our main result shows that if any one of these configurations is useless for the stated task, then so are the other two. Moreover, this is the case precisely when $\ket{\psi_{ABC}}$ is such that the reduced states $\rho_{AB}$ and $\sigma_{AC}$ are both separable.}
    \label{fig:my_label}
\end{figure}

\twocolumngrid

In order to prove the above results, we exploit the distillability properties of {\em{low rank}} quantum states, i.e., states $\rho_{AB}$ satisfying $\rank\rho_{AB}< \rank \rho_B$. It is known that such states are 2-way distillable \cite[Theorem 1]{Horodecki2003entanglement}. We first improve this result by deriving an explict lower bound on the 2-way distillable entanglement of such states (Theorem~\ref{theorem:main1}). Furthermore, we show that not all low rank states are 1-way distillable (Example~\ref{eg:distill}). However, it turns out that a simple additional rank constraint on such states makes distillation possible by using only 1-way classical communication (Theorem~\ref{theorem:main2}). Using this result, we prove that a randomly selected low rank state is \emph{almost surely} 1-way distillable (Theorem~\ref{theorem:main3}). Finally, we use Theorem~\ref{theorem:main2} along with some previously known results to derive our main results on full undistillability of quantum states (Theorems~\ref{theorem:main4} and \ref{theorem:main5}, and Corollary~\ref{coro:main1}).

\section{Results}

As mentioned before, any low rank quantum state is known to be 2-way distillable \cite[Theorem 1]{Horodecki2003entanglement}. The following theorem strengthens this result by providing a concrete lower bound on the 2-way distillable entanglement of such states. Our proof exploits the \emph{hashing bound} \cite{Bennett1996distillationQECC, Devetak2005distillation}, which shows that for any state $\rho_{AB}$, the \emph{coherent information} $I^c_{\to}(\rho_{AB}):= S(\rho_B) - S(\rho_{AB})$ is a 1-way achievable rate for distilling entanglement from $\rho_{AB}$ via the \emph{hashing protocol}, where $S(\rho):=-\Tr (\rho \log \rho)$ is the \emph{von Neumann entropy} function (all logarithms are taken with base 2).

\begin{theorem}\label{theorem:main1}
Let $\rho_{AB}$ be a quantum state satisfying $r = \rank \rho_{AB} < \rank \rho_B = r_B$. Then,
\begin{equation*}
    D_{\leftrightarrow}(\rho_{AB}) \geq \lambda^B_{\min} r_B [\log r_B - \log r]>0,
\end{equation*}
where $\lambda^B_{min}$ is the minimum positive eigenvalue of $\rho_B$. In particular, $\rho_{AB}$ is 2-way distillable. Similarly, if $r<\rank \rho_A = r_A$, then
\begin{equation*}
    D_{\leftrightarrow}(\rho_{AB}) \geq \lambda^A_{\min} r_A [\log r_A - \log r]>0.
\end{equation*}
\end{theorem}
\begin{proof}
We prove only the first part. The idea is for Bob to locally apply a filter (this is just a local measurement on Bob's subsystem) which succeeds with some non-zero probability in such a way that the filtered state $\rho'_{AB}$ satisfies $\rho'_B = \Pi_B/r_B$, where $\Pi_B$ projects onto the support of $\rho_B$ (note that Bob has to convey the filtering results to Alice via backward classical communication, so that they can discard the states for which the filtering fails). The two parties can then distill entanglement from $\rho'_{AB}$ at the desired rate via the hashing protocol. Further details of the proof can be found in Appendix~\ref{appen:th-main1}.
\end{proof}

Note that the strategy used to prove Theorem~\ref{theorem:main1} would fail if classical communication is only allowed from Alice to Bob, since then Bob would not be able to convey the results of the filtering procedure to Alice. In fact, one can explicitly construct a state which has low rank but is still 1-way undistillable, as we now do.

\begin{example}\label{eg:distill}
We use the construction from \cite{Cubitt2008degradable}. Let $\Phi:\M{d_A}\to \M{d_A}\oplus \M{d_{A'}}\simeq \M{d_B}$ be defined as $\Phi(\rho) = (\rho \oplus \Lambda(\rho))/2$, where $d_A = d_{A'}$ and $\Lambda:\M{d_A}\to \M{d_{A'}}$ is a channel with $\rank J_{AA'}(\Lambda)=d_A^2$. Then, $\rank J_{AB}(\Phi)=d_A^2+1$. Moreover, any channel $\Phi^c:\M{d_A}\to \M{d_C}$ that is complementary to $\Phi$ is \emph{antidegradable}, i.e., there exists a channel $\mathcal{N}:\M{d_B}\to \M{d_C}$ such that $\Phi^c = \mathcal{N}\circ \Phi$. Such channels are known to have zero 1-way quantum capacity $\mathcal{Q}_{\to}(\Phi^c)=0$ \cite{Smith2012incapacity}. The Choi state $J_{AC}(\Phi^c)$ then serves as our desired example, since
\begin{align*}
    \operatorname{rank}J_{AC}(\Phi^c) &\leq 2d_A < d_A^2+1 = \operatorname{rank} J_C(\Phi^c),
\end{align*}
but $D_{\to}(J_{AC}(\Phi^c)) \leq \mathcal{Q}_{\to}(\Phi^c)=0$. On the other hand, Theorem~\ref{theorem:main1} ensures that
\begin{align*}
    \mathcal{Q}_{\leftrightarrow}(\Phi^c) &\geq D_{\leftrightarrow}(J_{AC}(\Phi^c)) \\ &\geq \lambda^C_{\min} (d^2_A+1) [\log (d^2_A+1) - \log (2d_A)],
\end{align*}
where $\lambda^C_{min}$ is the minimum positive eigenvalue of $J_C(\Phi^c)$.
\end{example}

We have seen that although low rank states are 2-way distillable, they may not always be 1-way distillable.
However, a simple additional constraint on these low rank states makes distillation of entanglement possible with only 1-way classical communication. Before stating this result, let us note that for any (normalized) $\ket{\phi_A}\in \C{d_A}$,
\begin{equation}\label{eq:rankbound}
    \operatorname{rank}\rho^\phi_B \leq \min \{\operatorname{rank}\rho_{AB}, \operatorname{rank}\rho_B\},
\end{equation}
where $\rho^{\phi}_B := \operatorname{Tr}_A [(\ketbra{\phi_A}\otimes \iden_B)\rho_{AB}].$ The proof of this claim can be found in Appendix~\ref{appen:eq-rankbound} \footnote{In \cite{Chen2011entanglement}, a state $\rho_{AB}$ is said to have the \emph{right full-rank property} if there exists $\ket{\phi_A}$ which saturates the bound in Eq.~\eqref{eq:rankbound}}. 

\begin{theorem}\label{theorem:main2}
Let $\rho_{AB}$ be a quantum state satisfying $r = \rank \rho_{AB} < \rank \rho_B = r_B$. If there exists a $\ket{\phi_A}\in \C{d_A}$ (normalized) with $\operatorname{rank}\rho^{\phi}_B = \min \{r,r_B \}=r$, then $\rho_{AB}$ is 1-way distillable, i.e., $D_{\to}(\rho_{AB})>0$.
\end{theorem}
\begin{proof}
Instead of Bob, it is now Alice who performs the filtering procedure in exactly the same way as before. The two parties finally share the filtered state $\rho'_{AB}$ satisfying $\rho'_{A}=\Pi_A/r_A$, which they use to probabilistically implement a channel $\Phi_{\rho'}$ from $A\to B$. This channel can be shown to have positive quantum capacity using recently obtained results \cite{Siddhu2021logsingularity, Singh2022detecting}, which when combined with Eq.~\eqref{eq:cap-distill-bound2} yields the stated result. The details of the proof are presented in Appendix~\ref{appen:th-main2}.
\end{proof}

The above theorem is the primary technical contribution of our work. The filtering method described above along with the bound in Eq.~\eqref{eq:cap-distill-bound2} allow us to exploit the recently developed perturbative tools for detecting positive quantum capacities of quantum channels \cite{Siddhu2021logsingularity, Singh2022detecting} to study distillability of low rank quantum states.

In general, for a state $\rho_{AB}$, it is difficult to ascertain the existence of $\ket{\phi_A}\in \C{d_A}$ such that $\rho^\phi_B$ saturates the bound in Eq.~\eqref{eq:rankbound}. However, we can show that for a randomly sampled low rank state, such a $\ket{\phi_A}$ exists almost surely, which means that such a state is almost surely 1-way distillable. In order to state this result more precisely, we first outline the steps to randomly sample a state \cite{Zyczkowski2011random}.
\begin{itemize}
    \item Fix dimensions $d_A,d_B\in \mathbb{N}, (d_A,d_B>1)$. 
    \item For each $d_C\in \mathbb{N}$, let $\ket{\psi_{ABC}}$ be a pure state sampled according to the Haar measure on the unit sphere in $\C{d_A}\otimes \C{d_B}\otimes \C{d_C}$ (see \cite{Zyczkowski2011random}).
    \item Then, $\rho_{AB}=\operatorname{Tr}_C\ketbra{\psi_{ABC}}$ is our randomly selected state (denoted as $\rho_{AB}\sim \mu_{d_C}$).
\end{itemize}

\begin{theorem}\label{theorem:main3}
Let $\rho_{AB}\sim \mu_{d_C}$ with $d_C < d_B$ be a random low rank state. Then, $\rho_{AB}$ is almost surely 1-way distillable, i.e., $D_{\to}(\rho_{AB})>0$ almost surely.
\end{theorem}
\begin{proof}
See Appendix~\ref{appen:th-main3}
\end{proof}

The analogue of the above result for quantum channels was obtained recently in \cite{Singh2021coherent}. The distillability properties of low rank quantum states are summarized in Table~I.

\onecolumngrid

\begin{table}[H] \label{table:distill}
\centering
\begin{tabular}{|l|l|}
\hline 
$D_{\leftrightarrow}(\rho_{AB})\geq \lambda^B_{\min} r_B [\log r_B - \log r]>0. \quad$     &   \makecell{$D_{\to}(\rho_{AB})$ can be zero but \\ $D_{\to}(\rho_{AB})>0$ \emph{almost surely}.} \\[7pt]
\hline
\end{tabular}
\caption{Distillability of low rank quantum states $\rho_{AB}$ satisfying $r=\rank \rho_{AB} < \operatorname{rank}\rho_B=r_B$.}
\end{table} 

\twocolumngrid

Before proceeding further, the readers should get themselves familiarized with some previously known results on the entanglement properties of low rank quantum states, which are listed below. 

\begin{theorem} \cite[Theorem 1]{Horodecki2003entanglement}
Let $\rho_{AB}$ be a quantum state satisfying $\rank \rho_{AB} < \max\{ \rank \rho_A, \rank\rho_B \}$. Then, $\rho_{AB}$ is 2-way distillable.
\end{theorem}

\begin{theorem} \label{theorem:prev1}
Let $\rho_{AB}$ be a quantum state satisfying $\rank \rho_{AB} \leq \max \{\rank \rho_A, \rank \rho_B\}$. Then, the following are equivalent:
\begin{enumerate}
    \item $\rho_{AB}$ is separable.
    \item $\rho_{AB}$ is PPT \cite[Theorem 1]{Horodecki2000entanglement}.
    \item $D_{\leftrightarrow}(\rho_{AB})=0$ \cite[Theorem 10]{Chen2011entanglement}.
\end{enumerate}
\end{theorem}

\begin{theorem} \cite[Lemma 4]{Chen2011entanglement2} \label{theorem:prev2}
Let $\rho_{AB}$ be a quantum state such that its complement $\sigma_{AC}$ is separable. Then, $\rho_{AB}$ is separable if and only if $I^c_{\to}(\rho_{AB})=0$.
\end{theorem}

\begin{theorem}\cite[Theorem 8]{Chen2011entanglement} \label{theorem:prev3}
Let $\rho_{AB}$ be a quantum state such that for all $\ket{\phi_A}$,
\begin{equation}
\rank \rho^\phi_B < \min\{\rank\rho_{AB}, \rank\rho_B \}.
\end{equation}
Then, $\rho_{AB}$ is 2-way distillable.
\end{theorem}

Equipped with the necessary tools, we are now ready to state and prove our main results. Before proceeding further, let us recall that if two states $\rho_{AB}$ and $\sigma_{AC}$ are complementary to each other, then they admit a common purification $\ket{\psi_{ABC}}$:
\begin{equation*}
    \rho_{AB}=\operatorname{Tr}_C \ketbra{\psi_{ABC}}, \,\, \sigma_{AC}=\operatorname{Tr}_B \ketbra{\psi_{ABC}}.
\end{equation*}
The following rank equalities then easily follow from Schmidt decomposition:
\begin{equation*}
    \rank \rho_{AB} = \rank \sigma_C, \,\,     \rank \sigma_{AC} = \operatorname{rank} \rho_B.    
\end{equation*}

\begin{theorem}\label{theorem:main4}
Let $\rho_{AB}$ be a state with $D_{\to}(\sigma_{AC})=0$, where $\sigma_{AC}$ is complementary to $\rho_{AB}$. Then, the following are equivalent:
\begin{enumerate}
    \item $\rho_{AB}$ is separable. \vspace{-0.2cm}
  \item $\rho_{AB}$ is PPT.
    \vspace{-0.2cm}
   \item $D_{\leftrightarrow}(\rho_{AB})=0$.
\end{enumerate}
\end{theorem}
\begin{proof}
The forward implications $1\implies 2\implies 3$ are well-known. To show that $3\implies 1$, assume that $D_{\leftrightarrow}(\rho_{AB})=0$. Since $D_{\to}(\sigma_{AC})=0$,
Theorem~\ref{theorem:main2} shows that either
\begin{equation}\label{eq:bla}
 \rank \rho_{AB} = \rank \sigma_C
    \leq \rank \sigma_{AC} = \operatorname{rank} \rho_B,    
\end{equation}
or $\sigma_{AC}$ is such that for all $\ket{\phi_A}\in \C{d_A}$,
\begin{align}
     \rank \rho^\phi_B = \rank \sigma^{\phi}_C &< \min \{\rank \sigma_{C}, \rank \sigma_{AC} \} \\ 
     & = \min \{\rank \rho_{AB}, \rank\rho_B \}. \nonumber
\end{align}
However, the latter condition implies that $D_{\leftrightarrow}(\rho_{AB})>0$ (Theorem~\ref{theorem:prev3}), which leads to a contradiction. Hence, Eq.~\eqref{eq:bla} is the only possibility, in which case Theorem~\ref{theorem:prev1} shows that $\rho_{AB}$ is separable.
\end{proof}

\begin{remark}
It is pertinent to mention here that the equivalences stated in Theorem~\ref{theorem:main4} were also obtained in \cite[Theorem 2]{Chen2011entanglement2}, but under a much stronger assumption of $\sigma_{AC}$ being PPT.
\end{remark}

Note that a state $\rho_{AB}$ with $D_{\to}(\sigma_{AC})=0$ may be such that $D_{\to}(\rho_{AB})=0$ but $D_{\leftrightarrow}(\rho_{AB})>0$. Choi states of NPT \emph{self-complementary} channels (i.e., channels $\Phi$ with NPT Choi states which admit a complement $\Phi^c=\Phi$) are examples of such states. For instance, the qutrit \emph{Werner-Holevo} channel $\Phi_{\rm{WH}}:\M{3}\to \M{3}$ defined as
\begin{equation}
    \Phi_{\rm{WH}}(X) = \frac{\operatorname{Tr}(X)\iden_3-X^\top}{2}
\end{equation}
is NPT self-complementary \cite{Cubitt2008degradable}, so that $\mathcal{Q}_{\to}(\Phi_{\rm{WH}})=0$. However, Theorem~\ref{theorem:main4} shows that $\mathcal{Q}_{\leftrightarrow}(\Phi_{\rm{WH}})>0$. We are thus led to the following result.
\begin{corollary}\label{coro:main3}
For any NPT self-complementary channel $\Phi:\M{d_A}\to \M{d_B}$, $D_{\to}(J_{AB}(\Phi))=\mathcal{Q}_{\to}(\Phi)=0$ but $\mathcal{Q}_{\leftrightarrow}(\Phi)\geq D_{\leftrightarrow}(J_{AB}(\Phi))>0$.
\end{corollary}

We now show that if a state $\rho_{AB}$ is such that $D_{\leftrightarrow}(\sigma_{AC})=0$, then it is 1-way undistillable if and only if it is 2-way undistillable.

\begin{theorem}\label{theorem:main5}
Let $\rho_{AB}$ be a state with $D_{\leftrightarrow}(\sigma_{AC})=0$, where $\sigma_{AC}$ is complementary to $\rho_{AB}$. Then, the following are equivalent:
\begin{enumerate}
    \item $\rho_{AB}$ is separable. \vspace{-0.2cm}
    \item $\rho_{AB}$ is PPT.
    \vspace{-0.2cm}
    \item $D_{\leftrightarrow}(\rho_{AB})=0$.
    \vspace{-0.2cm}
    \item $D_{\to}(\rho_{AB})=0$.
\end{enumerate}
\end{theorem}
\begin{proof}
All the forward implications are well-known. Hence, it suffices to prove $4\implies 1$. So assume that $D_{\to}(\rho_{AB})=0$. Then, $\sigma_{AC}$ must be separable according to Theorem~\ref{theorem:main4}. Moreover, since $D_{\to}(\rho_{AB})=0=D_{\to}(\sigma_{AC})$, we infer from the hashing bound that
\begin{align}
    I^c_{\to}(\rho_{AB}) &\leq 0, \nonumber \\
    I_{\to}^c(\sigma_{AC}) = -I^c_{\to}(\rho_{AB}) &\leq 0.
\end{align}
Hence, $I^c_{\to}(\rho_{AB})=0$ and the desired conclusion follows from Theorem~\ref{theorem:prev2}.
\end{proof}

By using Theorem \ref{theorem:main4} and Theorem \ref{theorem:main5}, our primary result as stated in the Introduction can be derived with ease.

\begin{corollary}\label{coro:main1}
Let $\ket{\psi_{ABC}}$ be a tripartite pure state. Furthermore, let $\rho_{AB}=\operatorname{Tr}_C \ketbra{\psi_{ABC}}$ and $\sigma_{AC}=\operatorname{Tr}_B \ketbra{\psi_{ABC}}$.  Then, the following are equivalent:
\begin{enumerate}
    \item $\rho_{AB}$ and $\sigma_{AC}$ are separable.
    \vspace{-0.1cm}
    \item $\rho_{AB}$ and $\sigma_{AC}$ are PPT.
    \vspace{-0.1cm}
    \item $D^{\leftrightarrow}_{\leftrightarrow}(\psi_{ABC})=0$.
    \vspace{-0.1cm}
    \item $D^{\leftrightarrow}_{\to}(\psi_{ABC})=0$.
    \vspace{-0.1cm}
    \item $D^{\to}_{\leftrightarrow}(\psi_{ABC})=0.$
\end{enumerate}
\end{corollary}
\begin{proof}
The stated equivalences follow immediately from Theorems~\ref{theorem:main4} and \ref{theorem:main5}.
\end{proof}

The above corollary can be considered as a significant generalization of \cite[Corollary 6]{Singh2022bippt}, where it was shown that if $\rho_{AB}$ and its complement $\sigma_{AC}$ are both PPT, then both of them must be separable.

Finally, let us recall that every PPT state is $2$-way undistillable, but we do not know if the converse holds \cite{Shor2000evidence, Pankowski2010nptbound}. The following corollary is a partial resolution of this problem. We show that even if there exists an NPT 2-way undistillable state $\rho_{AB}$, Alice can still distill entanglement with Charlie by using the complementary state $\sigma_{AC}$.

\begin{corollary}
Let $\rho_{AB}$ be an NPT state. Then, either $\rho_{AB}$ or its complement $\sigma_{AC}$ is 2-way distillable.
\end{corollary}
\begin{proof}
Assume on the contrary that $D_{\leftrightarrow}(\rho_{AB})=0=D_{\leftrightarrow}(\sigma_{AC})$. Then, Corollary~\ref{coro:main1} shows that $\rho_{AB}$ and $\sigma_{AC}$ are both PPT, thus leading to a contradiction.
\end{proof}

\section{Conclusion}
In this paper, we have studied how different configurations of classical communication between three parties sharing a pure state $\ket{\psi_{ABC}}$ affect their ability to distill entanglement. Our main result shows that the properties of being $2, (1,2),$ and $(2,1)-$way fully undistillable are equivalent for $\ket{\psi_{ABC}}$. In other words, if Alice is already allowed two-way classical communication
with Bob (say), then allowing two-way communication between Alice and Charlie is no better than just forward classical communication from Alice to Charlie. Moreover, this happens precisely when $\ket{\psi_{ABC}}$ is such that both its reductions $\rho_{AB}$ and $\sigma_{AC}$ are PPT, which is further equivalent to both the reductions being separable. In addition to providing new insight into the role of classical communication in entanglement distillation, our results also shed light on the structure of undistillable quantum states.

Since the 2-way full undistillability of a tripartite pure state implies separability of its two bipartite reductions (cf.~Corollary \ref{coro:main1}), it means that in this case, these reductions are also useless from the perspective of private key distillation. Hence, in the tripartite distillation setting that we consider, the notions of being 2-way entanglement-undistillable and 2-way key-undistillable are equivalent, in contrast to the bipartite setting \cite{Horodecki2005private, Horodecki2009private}.

It would be interesting to study the tightness of the lower bounds on distillable entanglement obtained in Theorem~\ref{theorem:main1}, perhaps by comparing them with some well-known upper bounds, such as the Rains bound \cite{Rains2001bound} and squashed entanglement \cite{Christandl2004bound}. 

In the spirit of Theorem~\ref{theorem:main3}, it is worth investigating if a threshold value $d^{\to}_C$ of the dimension of the purifying system can be obtained such that for $d_C<d^{\to}_C$ (resp. $d_C>d^{\to}_C$), a random state $\rho_{AB}\sim \mu_{d_C}$ typically satisfies $D_{\to}(\rho_{AB})>0$ (resp. $D_{\to}(\rho_{AB})=0$). Identifying a similar threshold $d^{\leftrightarrow}_C$ for the 2-way distillable entanglement and comparing it with the known threshold $d^{\operatorname{PPT}}_C$ for the PPT property of bipartite states \cite{Auburn2012thresh} can yield valuable insights into the NPT bound entanglement problem. \\

\textit{Acknowledgements} We would like to thank Mark Wilde for helpful comments on an earlier version of this paper. 

\PRLsep

\bibliography{references}

\begin{thebibliography}{37}%
\makeatletter
\providecommand \@ifxundefined [1]{%
 \@ifx{#1\undefined}
}%
\providecommand \@ifnum [1]{%
 \ifnum #1\expandafter \@firstoftwo
 \else \expandafter \@secondoftwo
 \fi
}%
\providecommand \@ifx [1]{%
 \ifx #1\expandafter \@firstoftwo
 \else \expandafter \@secondoftwo
 \fi
}%
\providecommand \natexlab [1]{#1}%
\providecommand \enquote  [1]{``#1''}%
\providecommand \bibnamefont  [1]{#1}%
\providecommand \bibfnamefont [1]{#1}%
\providecommand \citenamefont [1]{#1}%
\providecommand \href@noop [0]{\@secondoftwo}%
\providecommand \href [0]{\begingroup \@sanitize@url \@href}%
\providecommand \@href[1]{\@@startlink{#1}\@@href}%
\providecommand \@@href[1]{\endgroup#1\@@endlink}%
\providecommand \@sanitize@url [0]{\catcode `\\12\catcode `\$12\catcode
  `\&12\catcode `\#12\catcode `\^12\catcode `\_12\catcode `\%12\relax}%
\providecommand \@@startlink[1]{}%
\providecommand \@@endlink[0]{}%
\providecommand \url  [0]{\begingroup\@sanitize@url \@url }%
\providecommand \@url [1]{\endgroup\@href {#1}{\urlprefix }}%
\providecommand \urlprefix  [0]{URL }%
\providecommand \Eprint [0]{\href }%
\providecommand \doibase [0]{https://doi.org/}%
\providecommand \selectlanguage [0]{\@gobble}%
\providecommand \bibinfo  [0]{\@secondoftwo}%
\providecommand \bibfield  [0]{\@secondoftwo}%
\providecommand \translation [1]{[#1]}%
\providecommand \BibitemOpen [0]{}%
\providecommand \bibitemStop [0]{}%
\providecommand \bibitemNoStop [0]{.\EOS\space}%
\providecommand \EOS [0]{\spacefactor3000\relax}%
\providecommand \BibitemShut  [1]{\csname bibitem#1\endcsname}%
\let\auto@bib@innerbib\@empty
\bibitem [{\citenamefont {Bennett}\ \emph {et~al.}(1993)\citenamefont
  {Bennett}, \citenamefont {Brassard}, \citenamefont {Cr\'epeau}, \citenamefont
  {Jozsa}, \citenamefont {Peres},\ and\ \citenamefont
  {Wootters}}]{Bennett1993teleportation}%
  \BibitemOpen
  \bibfield  {author} {\bibinfo {author} {\bibfnamefont {C.~H.}\ \bibnamefont
  {Bennett}}, \bibinfo {author} {\bibfnamefont {G.}~\bibnamefont {Brassard}},
  \bibinfo {author} {\bibfnamefont {C.}~\bibnamefont {Cr\'epeau}}, \bibinfo
  {author} {\bibfnamefont {R.}~\bibnamefont {Jozsa}}, \bibinfo {author}
  {\bibfnamefont {A.}~\bibnamefont {Peres}},\ and\ \bibinfo {author}
  {\bibfnamefont {W.~K.}\ \bibnamefont {Wootters}},\ }\bibfield  {title}
  {\bibinfo {title} {Teleporting an unknown quantum state via dual classical
  and einstein-podolsky-rosen channels},\ }\href
  {https://doi.org/10.1103/PhysRevLett.70.1895} {\bibfield  {journal} {\bibinfo
   {journal} {Phys. Rev. Lett.}\ }\textbf {\bibinfo {volume} {70}},\ \bibinfo
  {pages} {1895} (\bibinfo {year} {1993})}\BibitemShut {NoStop}%
\bibitem [{\citenamefont {Ekert}(1991)}]{Ekert1991crypto}%
  \BibitemOpen
  \bibfield  {author} {\bibinfo {author} {\bibfnamefont {A.~K.}\ \bibnamefont
  {Ekert}},\ }\bibfield  {title} {\bibinfo {title} {Quantum cryptography based
  on bell's theorem},\ }\href {https://doi.org/10.1103/PhysRevLett.67.661}
  {\bibfield  {journal} {\bibinfo  {journal} {Phys. Rev. Lett.}\ }\textbf
  {\bibinfo {volume} {67}},\ \bibinfo {pages} {661} (\bibinfo {year}
  {1991})}\BibitemShut {NoStop}%
\bibitem [{\citenamefont {Gisin}\ \emph {et~al.}(2002)\citenamefont {Gisin},
  \citenamefont {Ribordy}, \citenamefont {Tittel},\ and\ \citenamefont
  {Zbinden}}]{Gisin2002crypto}%
  \BibitemOpen
  \bibfield  {author} {\bibinfo {author} {\bibfnamefont {N.}~\bibnamefont
  {Gisin}}, \bibinfo {author} {\bibfnamefont {G.}~\bibnamefont {Ribordy}},
  \bibinfo {author} {\bibfnamefont {W.}~\bibnamefont {Tittel}},\ and\ \bibinfo
  {author} {\bibfnamefont {H.}~\bibnamefont {Zbinden}},\ }\bibfield  {title}
  {\bibinfo {title} {Quantum cryptography},\ }\href
  {https://doi.org/10.1103/RevModPhys.74.145} {\bibfield  {journal} {\bibinfo
  {journal} {Rev. Mod. Phys.}\ }\textbf {\bibinfo {volume} {74}},\ \bibinfo
  {pages} {145} (\bibinfo {year} {2002})}\BibitemShut {NoStop}%
\bibitem [{\citenamefont {Bennett}\ and\ \citenamefont
  {Wiesner}(1992)}]{Bennett1992superdense}%
  \BibitemOpen
  \bibfield  {author} {\bibinfo {author} {\bibfnamefont {C.~H.}\ \bibnamefont
  {Bennett}}\ and\ \bibinfo {author} {\bibfnamefont {S.~J.}\ \bibnamefont
  {Wiesner}},\ }\bibfield  {title} {\bibinfo {title} {Communication via one-
  and two-particle operators on einstein-podolsky-rosen states},\ }\href
  {https://doi.org/10.1103/PhysRevLett.69.2881} {\bibfield  {journal} {\bibinfo
   {journal} {Phys. Rev. Lett.}\ }\textbf {\bibinfo {volume} {69}},\ \bibinfo
  {pages} {2881} (\bibinfo {year} {1992})}\BibitemShut {NoStop}%
\bibitem [{\citenamefont {Bennett}\ \emph
  {et~al.}(1996{\natexlab{a}})\citenamefont {Bennett}, \citenamefont
  {Brassard}, \citenamefont {Popescu}, \citenamefont {Schumacher},
  \citenamefont {Smolin},\ and\ \citenamefont
  {Wootters}}]{Bennett1996purification}%
  \BibitemOpen
  \bibfield  {author} {\bibinfo {author} {\bibfnamefont {C.~H.}\ \bibnamefont
  {Bennett}}, \bibinfo {author} {\bibfnamefont {G.}~\bibnamefont {Brassard}},
  \bibinfo {author} {\bibfnamefont {S.}~\bibnamefont {Popescu}}, \bibinfo
  {author} {\bibfnamefont {B.}~\bibnamefont {Schumacher}}, \bibinfo {author}
  {\bibfnamefont {J.~A.}\ \bibnamefont {Smolin}},\ and\ \bibinfo {author}
  {\bibfnamefont {W.~K.}\ \bibnamefont {Wootters}},\ }\bibfield  {title}
  {\bibinfo {title} {Purification of noisy entanglement and faithful
  teleportation via noisy channels},\ }\href
  {https://doi.org/10.1103/PhysRevLett.76.722} {\bibfield  {journal} {\bibinfo
  {journal} {Phys. Rev. Lett.}\ }\textbf {\bibinfo {volume} {76}},\ \bibinfo
  {pages} {722} (\bibinfo {year} {1996}{\natexlab{a}})}\BibitemShut {NoStop}%
\bibitem [{\citenamefont {Bennett}\ \emph
  {et~al.}(1996{\natexlab{b}})\citenamefont {Bennett}, \citenamefont
  {DiVincenzo}, \citenamefont {Smolin},\ and\ \citenamefont
  {Wootters}}]{Bennett1996distillationQECC}%
  \BibitemOpen
  \bibfield  {author} {\bibinfo {author} {\bibfnamefont {C.~H.}\ \bibnamefont
  {Bennett}}, \bibinfo {author} {\bibfnamefont {D.~P.}\ \bibnamefont
  {DiVincenzo}}, \bibinfo {author} {\bibfnamefont {J.~A.}\ \bibnamefont
  {Smolin}},\ and\ \bibinfo {author} {\bibfnamefont {W.~K.}\ \bibnamefont
  {Wootters}},\ }\bibfield  {title} {\bibinfo {title} {Mixed-state entanglement
  and quantum error correction},\ }\href
  {https://doi.org/10.1103/PhysRevA.54.3824} {\bibfield  {journal} {\bibinfo
  {journal} {Phys. Rev. A}\ }\textbf {\bibinfo {volume} {54}},\ \bibinfo
  {pages} {3824} (\bibinfo {year} {1996}{\natexlab{b}})}\BibitemShut {NoStop}%
\bibitem [{\citenamefont {Horodecki}\ and\ \citenamefont
  {Horodecki}(2001)}]{Horodecki2001review}%
  \BibitemOpen
  \bibfield  {author} {\bibinfo {author} {\bibfnamefont {P.}~\bibnamefont
  {Horodecki}}\ and\ \bibinfo {author} {\bibfnamefont {R.}~\bibnamefont
  {Horodecki}},\ }\bibfield  {title} {\bibinfo {title} {Distillation and bound
  entanglement},\ }\href@noop {} {\bibfield  {journal} {\bibinfo  {journal}
  {Quantum Info. Comput.}\ }\textbf {\bibinfo {volume} {1}},\ \bibinfo {pages}
  {45–75} (\bibinfo {year} {2001})}\BibitemShut {NoStop}%
\bibitem [{Note1()}]{Note1}%
  \BibitemOpen
  \bibinfo {note} {The isometric freedom in choosing a purification $\mathinner
  {|{\psi _{ABC}}\rangle }$ of $\rho _{AB}$ imparts a similar freedom in the
  choice of a complementary state $\sigma _{AC}$. However, it is easy to check
  that different choices of $\sigma _{AC}$ yield the same value of distillable
  entanglement $D_{\to }(\sigma _{AC})$ or $D_{\leftrightarrow }(\sigma
  _{AC})$.}\BibitemShut {Stop}%
\bibitem [{\citenamefont {Coffman}\ \emph {et~al.}(2000)\citenamefont
  {Coffman}, \citenamefont {Kundu},\ and\ \citenamefont
  {Wootters}}]{Coffman2000monogamy}%
  \BibitemOpen
  \bibfield  {author} {\bibinfo {author} {\bibfnamefont {V.}~\bibnamefont
  {Coffman}}, \bibinfo {author} {\bibfnamefont {J.}~\bibnamefont {Kundu}},\
  and\ \bibinfo {author} {\bibfnamefont {W.~K.}\ \bibnamefont {Wootters}},\
  }\bibfield  {title} {\bibinfo {title} {Distributed entanglement},\ }\href
  {https://doi.org/10.1103/PhysRevA.61.052306} {\bibfield  {journal} {\bibinfo
  {journal} {Phys. Rev. A}\ }\textbf {\bibinfo {volume} {61}},\ \bibinfo
  {pages} {052306} (\bibinfo {year} {2000})}\BibitemShut {NoStop}%
\bibitem [{\citenamefont {Osborne}\ and\ \citenamefont
  {Verstraete}(2006)}]{Osborne2006monogamy}%
  \BibitemOpen
  \bibfield  {author} {\bibinfo {author} {\bibfnamefont {T.~J.}\ \bibnamefont
  {Osborne}}\ and\ \bibinfo {author} {\bibfnamefont {F.}~\bibnamefont
  {Verstraete}},\ }\bibfield  {title} {\bibinfo {title} {General monogamy
  inequality for bipartite qubit entanglement},\ }\href
  {https://doi.org/10.1103/PhysRevLett.96.220503} {\bibfield  {journal}
  {\bibinfo  {journal} {Phys. Rev. Lett.}\ }\textbf {\bibinfo {volume} {96}},\
  \bibinfo {pages} {220503} (\bibinfo {year} {2006})}\BibitemShut {NoStop}%
\bibitem [{\citenamefont {Smolin}\ \emph {et~al.}(2005)\citenamefont {Smolin},
  \citenamefont {Verstraete},\ and\ \citenamefont
  {Winter}}]{Smolin2005assistance}%
  \BibitemOpen
  \bibfield  {author} {\bibinfo {author} {\bibfnamefont {J.~A.}\ \bibnamefont
  {Smolin}}, \bibinfo {author} {\bibfnamefont {F.}~\bibnamefont {Verstraete}},\
  and\ \bibinfo {author} {\bibfnamefont {A.}~\bibnamefont {Winter}},\
  }\bibfield  {title} {\bibinfo {title} {Entanglement of assistance and
  multipartite state distillation},\ }\bibfield  {journal} {\bibinfo  {journal}
  {Physical Review A}\ }\textbf {\bibinfo {volume} {72}},\ \href
  {https://doi.org/10.1103/physreva.72.052317} {10.1103/physreva.72.052317}
  (\bibinfo {year} {2005})\BibitemShut {NoStop}%
\bibitem [{\citenamefont {Wolf}(2012)}]{Wolf2012Qtour}%
  \BibitemOpen
  \bibfield  {author} {\bibinfo {author} {\bibfnamefont {M.~M.}\ \bibnamefont
  {Wolf}},\ }\bibfield  {title} {\bibinfo {title} {Quantum channels and
  operations: Guided tour},\ }\href@noop {} {\bibfield  {journal} {\bibinfo
  {journal} {(unpublished)}\ } (\bibinfo {year} {2012})}\BibitemShut {NoStop}%
\bibitem [{\citenamefont {Horodecki}\ \emph {et~al.}(1996)\citenamefont
  {Horodecki}, \citenamefont {Horodecki},\ and\ \citenamefont
  {Horodecki}}]{Horodecki1996sep}%
  \BibitemOpen
  \bibfield  {author} {\bibinfo {author} {\bibfnamefont {M.}~\bibnamefont
  {Horodecki}}, \bibinfo {author} {\bibfnamefont {P.}~\bibnamefont
  {Horodecki}},\ and\ \bibinfo {author} {\bibfnamefont {R.}~\bibnamefont
  {Horodecki}},\ }\bibfield  {title} {\bibinfo {title} {Separability of mixed
  states: necessary and sufficient conditions},\ }\href
  {https://doi.org/10.1016/s0375-9601(96)00706-2} {\bibfield  {journal}
  {\bibinfo  {journal} {Physics Letters A}\ }\textbf {\bibinfo {volume}
  {223}},\ \bibinfo {pages} {1} (\bibinfo {year} {1996})}\BibitemShut {NoStop}%
\bibitem [{\citenamefont {DiVincenzo}\ \emph {et~al.}(2000)\citenamefont
  {DiVincenzo}, \citenamefont {Shor}, \citenamefont {Smolin}, \citenamefont
  {Terhal},\ and\ \citenamefont {Thapliyal}}]{Shor2000evidence}%
  \BibitemOpen
  \bibfield  {author} {\bibinfo {author} {\bibfnamefont {D.~P.}\ \bibnamefont
  {DiVincenzo}}, \bibinfo {author} {\bibfnamefont {P.~W.}\ \bibnamefont
  {Shor}}, \bibinfo {author} {\bibfnamefont {J.~A.}\ \bibnamefont {Smolin}},
  \bibinfo {author} {\bibfnamefont {B.~M.}\ \bibnamefont {Terhal}},\ and\
  \bibinfo {author} {\bibfnamefont {A.~V.}\ \bibnamefont {Thapliyal}},\
  }\bibfield  {title} {\bibinfo {title} {Evidence for bound entangled states
  with negative partial transpose},\ }\href
  {https://doi.org/10.1103/PhysRevA.61.062312} {\bibfield  {journal} {\bibinfo
  {journal} {Phys. Rev. A}\ }\textbf {\bibinfo {volume} {61}},\ \bibinfo
  {pages} {062312} (\bibinfo {year} {2000})}\BibitemShut {NoStop}%
\bibitem [{\citenamefont {Pankowski}\ \emph {et~al.}(2010)\citenamefont
  {Pankowski}, \citenamefont {Piani}, \citenamefont {Horodecki},\ and\
  \citenamefont {Horodecki}}]{Pankowski2010nptbound}%
  \BibitemOpen
  \bibfield  {author} {\bibinfo {author} {\bibfnamefont {{\L}.}~\bibnamefont
  {Pankowski}}, \bibinfo {author} {\bibfnamefont {M.}~\bibnamefont {Piani}},
  \bibinfo {author} {\bibfnamefont {M.}~\bibnamefont {Horodecki}},\ and\
  \bibinfo {author} {\bibfnamefont {P.}~\bibnamefont {Horodecki}},\ }\bibfield
  {title} {\bibinfo {title} {A few steps more towards {NPT} bound
  entanglement},\ }\href {https://doi.org/10.1109/tit.2010.2050810} {\bibfield
  {journal} {\bibinfo  {journal} {{IEEE} Transactions on Information Theory}\
  }\textbf {\bibinfo {volume} {56}},\ \bibinfo {pages} {4085} (\bibinfo {year}
  {2010})}\BibitemShut {NoStop}%
\bibitem [{\citenamefont {Horodecki}\ \emph {et~al.}(2003)\citenamefont
  {Horodecki}, \citenamefont {Smolin}, \citenamefont {Terhal},\ and\
  \citenamefont {Thapliyal}}]{Horodecki2003entanglement}%
  \BibitemOpen
  \bibfield  {author} {\bibinfo {author} {\bibfnamefont {P.}~\bibnamefont
  {Horodecki}}, \bibinfo {author} {\bibfnamefont {J.~A.}\ \bibnamefont
  {Smolin}}, \bibinfo {author} {\bibfnamefont {B.~M.}\ \bibnamefont {Terhal}},\
  and\ \bibinfo {author} {\bibfnamefont {A.~V.}\ \bibnamefont {Thapliyal}},\
  }\bibfield  {title} {\bibinfo {title} {Rank two bipartite bound entangled
  states do not exist},\ }\href {https://doi.org/10.1016/s0304-3975(01)00376-0}
  {\bibfield  {journal} {\bibinfo  {journal} {Theoretical Computer Science}\
  }\textbf {\bibinfo {volume} {292}},\ \bibinfo {pages} {589} (\bibinfo {year}
  {2003})}\BibitemShut {NoStop}%
\bibitem [{\citenamefont {Devetak}\ and\ \citenamefont
  {Winter}(2005)}]{Devetak2005distillation}%
  \BibitemOpen
  \bibfield  {author} {\bibinfo {author} {\bibfnamefont {I.}~\bibnamefont
  {Devetak}}\ and\ \bibinfo {author} {\bibfnamefont {A.}~\bibnamefont
  {Winter}},\ }\bibfield  {title} {\bibinfo {title} {Distillation of secret key
  and entanglement from quantum states},\ }\href
  {https://doi.org/10.1098/rspa.2004.1372} {\bibfield  {journal} {\bibinfo
  {journal} {Proceedings of the Royal Society A: Mathematical, Physical and
  Engineering Sciences}\ }\textbf {\bibinfo {volume} {461}},\ \bibinfo {pages}
  {207} (\bibinfo {year} {2005})}\BibitemShut {NoStop}%
\bibitem [{\citenamefont {Cubitt}\ \emph {et~al.}(2008)\citenamefont {Cubitt},
  \citenamefont {Ruskai},\ and\ \citenamefont {Smith}}]{Cubitt2008degradable}%
  \BibitemOpen
  \bibfield  {author} {\bibinfo {author} {\bibfnamefont {T.~S.}\ \bibnamefont
  {Cubitt}}, \bibinfo {author} {\bibfnamefont {M.~B.}\ \bibnamefont {Ruskai}},\
  and\ \bibinfo {author} {\bibfnamefont {G.}~\bibnamefont {Smith}},\ }\bibfield
   {title} {\bibinfo {title} {The structure of degradable quantum channels},\
  }\href {https://doi.org/10.1063/1.2953685} {\bibfield  {journal} {\bibinfo
  {journal} {J. Math. Phys.}\ }\textbf {\bibinfo {volume} {49}},\ \bibinfo
  {pages} {102104} (\bibinfo {year} {2008})}\BibitemShut {NoStop}%
\bibitem [{\citenamefont {Smith}\ and\ \citenamefont
  {Smolin}(2012)}]{Smith2012incapacity}%
  \BibitemOpen
  \bibfield  {author} {\bibinfo {author} {\bibfnamefont {G.}~\bibnamefont
  {Smith}}\ and\ \bibinfo {author} {\bibfnamefont {J.~A.}\ \bibnamefont
  {Smolin}},\ }\bibfield  {title} {\bibinfo {title} {Detecting incapacity of a
  quantum channel},\ }\href {https://doi.org/10.1103/PhysRevLett.108.230507}
  {\bibfield  {journal} {\bibinfo  {journal} {Phys. Rev. Lett.}\ }\textbf
  {\bibinfo {volume} {108}},\ \bibinfo {pages} {230507} (\bibinfo {year}
  {2012})}\BibitemShut {NoStop}%
\bibitem [{Note2()}]{Note2}%
  \BibitemOpen
  \bibinfo {note} {In \cite {Chen2011entanglement}, a state $\rho _{AB}$ is
  said to have the \protect \emph {right full-rank property} if there exists
  $\mathinner {|{\phi _A}\rangle }$ which saturates the bound in Eq.~\protect
  \textup {\hbox {\mathsurround \z@ \protect \normalfont (\ignorespaces \ref
  {eq:rankbound}\unskip \@@italiccorr )}}}\BibitemShut {NoStop}%
\bibitem [{\citenamefont {Siddhu}(2021)}]{Siddhu2021logsingularity}%
  \BibitemOpen
  \bibfield  {author} {\bibinfo {author} {\bibfnamefont {V.}~\bibnamefont
  {Siddhu}},\ }\bibfield  {title} {\bibinfo {title} {Entropic singularities
  give rise to quantum transmission},\ }\bibfield  {journal} {\bibinfo
  {journal} {Nature Communications}\ }\textbf {\bibinfo {volume} {12}},\ \href
  {https://doi.org/10.1038/s41467-021-25954-0} {10.1038/s41467-021-25954-0}
  (\bibinfo {year} {2021})\BibitemShut {NoStop}%
\bibitem [{\citenamefont {Singh}\ and\ \citenamefont
  {Datta}(2022{\natexlab{a}})}]{Singh2022detecting}%
  \BibitemOpen
  \bibfield  {author} {\bibinfo {author} {\bibfnamefont {S.}~\bibnamefont
  {Singh}}\ and\ \bibinfo {author} {\bibfnamefont {N.}~\bibnamefont {Datta}},\
  }\bibfield  {title} {\bibinfo {title} {Detecting positive quantum capacities
  of quantum channels},\ }\bibfield  {journal} {\bibinfo  {journal} {npj
  Quantum Information}\ }\textbf {\bibinfo {volume} {8}},\ \href
  {https://doi.org/10.1038/s41534-022-00550-2} {10.1038/s41534-022-00550-2}
  (\bibinfo {year} {2022}{\natexlab{a}})\BibitemShut {NoStop}%
\bibitem [{\citenamefont {{\.{Z}}yczkowski}\ \emph {et~al.}(2011)\citenamefont
  {{\.{Z}}yczkowski}, \citenamefont {Penson}, \citenamefont {Nechita},\ and\
  \citenamefont {Collins}}]{Zyczkowski2011random}%
  \BibitemOpen
  \bibfield  {author} {\bibinfo {author} {\bibfnamefont {K.}~\bibnamefont
  {{\.{Z}}yczkowski}}, \bibinfo {author} {\bibfnamefont {K.~A.}\ \bibnamefont
  {Penson}}, \bibinfo {author} {\bibfnamefont {I.}~\bibnamefont {Nechita}},\
  and\ \bibinfo {author} {\bibfnamefont {B.}~\bibnamefont {Collins}},\
  }\bibfield  {title} {\bibinfo {title} {Generating random density matrices},\
  }\href {https://doi.org/10.1063/1.3595693} {\bibfield  {journal} {\bibinfo
  {journal} {Journal of Mathematical Physics}\ }\textbf {\bibinfo {volume}
  {52}},\ \bibinfo {pages} {062201} (\bibinfo {year} {2011})}\BibitemShut
  {NoStop}%
\bibitem [{\citenamefont {Singh}\ and\ \citenamefont
  {Datta}(2022{\natexlab{b}})}]{Singh2021coherent}%
  \BibitemOpen
  \bibfield  {author} {\bibinfo {author} {\bibfnamefont {S.}~\bibnamefont
  {Singh}}\ and\ \bibinfo {author} {\bibfnamefont {N.}~\bibnamefont {Datta}},\
  }\bibfield  {title} {\bibinfo {title} {Coherent information of a quantum
  channel or its complement is generically positive},\ }\href
  {https://doi.org/10.22331/q-2022-08-11-775} {\bibfield  {journal} {\bibinfo
  {journal} {{Quantum}}\ }\textbf {\bibinfo {volume} {6}},\ \bibinfo {pages}
  {775} (\bibinfo {year} {2022}{\natexlab{b}})}\BibitemShut {NoStop}%
\bibitem [{\citenamefont {Horodecki}\ \emph {et~al.}(2000)\citenamefont
  {Horodecki}, \citenamefont {Lewenstein}, \citenamefont {Vidal},\ and\
  \citenamefont {Cirac}}]{Horodecki2000entanglement}%
  \BibitemOpen
  \bibfield  {author} {\bibinfo {author} {\bibfnamefont {P.}~\bibnamefont
  {Horodecki}}, \bibinfo {author} {\bibfnamefont {M.}~\bibnamefont
  {Lewenstein}}, \bibinfo {author} {\bibfnamefont {G.}~\bibnamefont {Vidal}},\
  and\ \bibinfo {author} {\bibfnamefont {I.}~\bibnamefont {Cirac}},\ }\bibfield
   {title} {\bibinfo {title} {Operational criterion and constructive checks for
  the separability of low-rank density matrices},\ }\href
  {https://doi.org/10.1103/PhysRevA.62.032310} {\bibfield  {journal} {\bibinfo
  {journal} {Phys. Rev. A}\ }\textbf {\bibinfo {volume} {62}},\ \bibinfo
  {pages} {032310} (\bibinfo {year} {2000})}\BibitemShut {NoStop}%
\bibitem [{\citenamefont {Chen}\ and\ \citenamefont
  {{D}okovi{\'{c}}}(2011)}]{Chen2011entanglement}%
  \BibitemOpen
  \bibfield  {author} {\bibinfo {author} {\bibfnamefont {L.}~\bibnamefont
  {Chen}}\ and\ \bibinfo {author} {\bibfnamefont {D.~{\v{Z}}.}\ \bibnamefont
  {{D}okovi{\'{c}}}},\ }\bibfield  {title} {\bibinfo {title} {Distillability
  and {PPT} entanglement of low-rank quantum states},\ }\href
  {https://doi.org/10.1088/1751-8113/44/28/285303} {\bibfield  {journal}
  {\bibinfo  {journal} {Journal of Physics A: Mathematical and Theoretical}\
  }\textbf {\bibinfo {volume} {44}},\ \bibinfo {pages} {285303} (\bibinfo
  {year} {2011})}\BibitemShut {NoStop}%
\bibitem [{\citenamefont {Hayashi}\ and\ \citenamefont
  {Chen}(2011)}]{Chen2011entanglement2}%
  \BibitemOpen
  \bibfield  {author} {\bibinfo {author} {\bibfnamefont {M.}~\bibnamefont
  {Hayashi}}\ and\ \bibinfo {author} {\bibfnamefont {L.}~\bibnamefont {Chen}},\
  }\bibfield  {title} {\bibinfo {title} {Weaker entanglement between two
  parties guarantees stronger entanglement with a third party},\ }\href
  {https://doi.org/10.1103/PhysRevA.84.012325} {\bibfield  {journal} {\bibinfo
  {journal} {Phys. Rev. A}\ }\textbf {\bibinfo {volume} {84}},\ \bibinfo
  {pages} {012325} (\bibinfo {year} {2011})}\BibitemShut {NoStop}%
\bibitem [{\citenamefont {Müller-Hermes}\ and\ \citenamefont
  {Singh}(2022)}]{Singh2022bippt}%
  \BibitemOpen
  \bibfield  {author} {\bibinfo {author} {\bibfnamefont {A.}~\bibnamefont
  {Müller-Hermes}}\ and\ \bibinfo {author} {\bibfnamefont {S.}~\bibnamefont
  {Singh}},\ }\bibfield  {title} {\bibinfo {title} {Bi-{PPT} channels are
  entanglement breaking},\ }\href@noop {} {\bibfield  {journal} {\bibinfo
  {journal} {arXiv:2204.01685}\ } (\bibinfo {year} {2022})}\BibitemShut
  {NoStop}%
\bibitem [{\citenamefont {Horodecki}\ \emph {et~al.}(2005)\citenamefont
  {Horodecki}, \citenamefont {Horodecki}, \citenamefont {Horodecki},\ and\
  \citenamefont {Oppenheim}}]{Horodecki2005private}%
  \BibitemOpen
  \bibfield  {author} {\bibinfo {author} {\bibfnamefont {K.}~\bibnamefont
  {Horodecki}}, \bibinfo {author} {\bibfnamefont {M.}~\bibnamefont
  {Horodecki}}, \bibinfo {author} {\bibfnamefont {P.}~\bibnamefont
  {Horodecki}},\ and\ \bibinfo {author} {\bibfnamefont {J.}~\bibnamefont
  {Oppenheim}},\ }\bibfield  {title} {\bibinfo {title} {Secure key from bound
  entanglement},\ }\href {https://doi.org/10.1103/PhysRevLett.94.160502}
  {\bibfield  {journal} {\bibinfo  {journal} {Phys. Rev. Lett.}\ }\textbf
  {\bibinfo {volume} {94}},\ \bibinfo {pages} {160502} (\bibinfo {year}
  {2005})}\BibitemShut {NoStop}%
\bibitem [{\citenamefont {Horodecki}\ \emph {et~al.}(2009)\citenamefont
  {Horodecki}, \citenamefont {Horodecki}, \citenamefont {Horodecki},\ and\
  \citenamefont {Oppenheim}}]{Horodecki2009private}%
  \BibitemOpen
  \bibfield  {author} {\bibinfo {author} {\bibfnamefont {K.}~\bibnamefont
  {Horodecki}}, \bibinfo {author} {\bibfnamefont {M.}~\bibnamefont
  {Horodecki}}, \bibinfo {author} {\bibfnamefont {P.}~\bibnamefont
  {Horodecki}},\ and\ \bibinfo {author} {\bibfnamefont {J.}~\bibnamefont
  {Oppenheim}},\ }\bibfield  {title} {\bibinfo {title} {General paradigm for
  distilling classical key from quantum states},\ }\href
  {https://doi.org/10.1109/tit.2008.2009798} {\bibfield  {journal} {\bibinfo
  {journal} {{IEEE} Transactions on Information Theory}\ }\textbf {\bibinfo
  {volume} {55}},\ \bibinfo {pages} {1898} (\bibinfo {year}
  {2009})}\BibitemShut {NoStop}%
\bibitem [{\citenamefont {Rains}(2001)}]{Rains2001bound}%
  \BibitemOpen
  \bibfield  {author} {\bibinfo {author} {\bibfnamefont {E.}~\bibnamefont
  {Rains}},\ }\bibfield  {title} {\bibinfo {title} {A semidefinite program for
  distillable entanglement},\ }\href {https://doi.org/10.1109/18.959270}
  {\bibfield  {journal} {\bibinfo  {journal} {{IEEE} Transactions on
  Information Theory}\ }\textbf {\bibinfo {volume} {47}},\ \bibinfo {pages}
  {2921} (\bibinfo {year} {2001})}\BibitemShut {NoStop}%
\bibitem [{\citenamefont {Christandl}\ and\ \citenamefont
  {Winter}(2004)}]{Christandl2004bound}%
  \BibitemOpen
  \bibfield  {author} {\bibinfo {author} {\bibfnamefont {M.}~\bibnamefont
  {Christandl}}\ and\ \bibinfo {author} {\bibfnamefont {A.}~\bibnamefont
  {Winter}},\ }\bibfield  {title} {\bibinfo {title}
  {{\textquotedblleft}squashed entanglement{\textquotedblright}: An additive
  entanglement measure},\ }\href {https://doi.org/10.1063/1.1643788} {\bibfield
   {journal} {\bibinfo  {journal} {Journal of Mathematical Physics}\ }\textbf
  {\bibinfo {volume} {45}},\ \bibinfo {pages} {829} (\bibinfo {year}
  {2004})}\BibitemShut {NoStop}%
\bibitem [{\citenamefont {Aubrun}(2012)}]{Auburn2012thresh}%
  \BibitemOpen
  \bibfield  {author} {\bibinfo {author} {\bibfnamefont {G.}~\bibnamefont
  {Aubrun}},\ }\bibfield  {title} {\bibinfo {title} {Partial transposition of
  random states and non-centered semicircular distributions},\ }\href
  {https://doi.org/10.1142/s2010326312500013} {\bibfield  {journal} {\bibinfo
  {journal} {Random Matrices: Theory and Applications}\ }\textbf {\bibinfo
  {volume} {01}},\ \bibinfo {pages} {1250001} (\bibinfo {year}
  {2012})}\BibitemShut {NoStop}%
\bibitem [{Note3()}]{Note3}%
  \BibitemOpen
  \bibinfo {note} {It is easy to show that $\Phi $ and $\Phi _c$ are
  complementary if and only if their Choi states are complementary}\BibitemShut
  {NoStop}%
\bibitem [{\citenamefont {Lloyd}(1997)}]{Lloyd1997capacity}%
  \BibitemOpen
  \bibfield  {author} {\bibinfo {author} {\bibfnamefont {S.}~\bibnamefont
  {Lloyd}},\ }\bibfield  {title} {\bibinfo {title} {Capacity of the noisy
  quantum channel},\ }\href {https://doi.org/10.1103/PhysRevA.55.1613}
  {\bibfield  {journal} {\bibinfo  {journal} {Phys. Rev. A}\ }\textbf {\bibinfo
  {volume} {55}},\ \bibinfo {pages} {1613} (\bibinfo {year}
  {1997})}\BibitemShut {NoStop}%
\bibitem [{\citenamefont {Devetak}(2005)}]{Devetak2005capacity}%
  \BibitemOpen
  \bibfield  {author} {\bibinfo {author} {\bibfnamefont {I.}~\bibnamefont
  {Devetak}},\ }\bibfield  {title} {\bibinfo {title} {The private classical
  capacity and quantum capacity of a quantum channel},\ }\href
  {https://doi.org/10.1109/TIT.2004.839515} {\bibfield  {journal} {\bibinfo
  {journal} {IEEE Transactions on Information Theory}\ }\textbf {\bibinfo
  {volume} {51}},\ \bibinfo {pages} {44} (\bibinfo {year} {2005})}\BibitemShut
  {NoStop}%
\bibitem [{\citenamefont {Nechita}(2007)}]{Nechita2007random}%
  \BibitemOpen
  \bibfield  {author} {\bibinfo {author} {\bibfnamefont {I.}~\bibnamefont
  {Nechita}},\ }\bibfield  {title} {\bibinfo {title} {Asymptotics of random
  density matrices},\ }\href {https://doi.org/10.1007/s00023-007-0345-5}
  {\bibfield  {journal} {\bibinfo  {journal} {Annales Henri Poincar{\'{e}}}\
  }\textbf {\bibinfo {volume} {8}},\ \bibinfo {pages} {1521} (\bibinfo {year}
  {2007})}\BibitemShut {NoStop}%
\end{thebibliography}%


\appendix

\section{Coherent information and the Proof of Eq.~\eqref{eq:cap-distill-bound2}}  \label{appen:capbound}

As in the case of entanglement distillation, the leakage of information by a channel $\Phi:\M{d_A}\to \M{d_B}$ to the `environment' is modelled via a third party called Charlie. Stinespring dilation theorem shows that there exists an isometry $V:\C{d_A}\to \C{d_B}\otimes \C{d_C}$ such that Bob receives the output $\Phi(\rho)=\operatorname{Tr}_C(V\rho V^\dagger)$ while Charlie receives the \emph{complementary} output $\Phi^c(\rho)=\operatorname{Tr}_B(V\rho V^\dagger)$ \footnote{It is easy to show that $\Phi$ and $\Phi_c$ are complementary if and only if their Choi states are complementary}. The \emph{coherent information} of $\Phi: \M{d_{A}}\to \M{d_B}$ is defined as 
\begin{align}
    \mathcal{Q}^{(1)}(\Phi) &:= \max_{\rho_{A}} \big( \, S[\Phi (\rho_{A})] - S[\Phi^c(\rho_{A})] \, \big) \label{eq:coherent}  \\
    &= \max_{\ket{\psi_{AA'}}} I^c_{\to}(\sigma_{A'B}),
\end{align}
where 
\begin{equation}\label{eq:sigma}
\sigma_{A'B} = (\Phi_{A\to B} \otimes \mathrm{id}_{A'\to A'}) (\ketbra{\psi_{AA'}})    
\end{equation}
and the optimization is over all pure states $\ket{\psi_{AA'}}\in \C{d_A}\otimes \C{d_{A'}}$ (or over all mixed states $\rho_{A}$ in Eq.~\eqref{eq:coherent}) with $d_A = d_{A'}$. It is known \cite{Lloyd1997capacity,Shor2000evidence, Devetak2005capacity} that the full 1-way capacity of $\Phi$ admits the following regularized expression:
\begin{equation}
    \mathcal{Q}_{\to}(\Phi) = \mathcal{Q}(\Phi) = \lim_{n\to \infty} \frac{\mathcal{Q}^{(1)}(\Phi^{\otimes n})}{n}. 
\end{equation}

\textit{Proof of Eq.~\eqref{eq:cap-distill-bound2}}. The crucial step in this proof is a probabilistic implementation of $\Phi$ via its Choi state with the help of teleportation. Assume that Alice and Bob share $N$ copies of the Choi state $J_{AB}(\Phi)$. Alice then locally prepares $N$ copies of a pure state $\ket{\psi_{AA'}}$ with $d_A=d_{A'}$. On each copy of $J_{AB}(\Phi)$ and $\ket{\psi_{AA'}}$, Alice performs a generalized measurement on the two $A$ systems with measurement operators $\{\Omega^+_{d_A}, \iden - \Omega^+_{d_A}\}$. It is easy to see that Alice obtains the outcome corresponding to the maximally entangled state $\Omega^+_{d_A}$ with probability $1/d^2_{A}$, in which case Alice and Bob finally share the state $\sigma_{A'B}$ as defined in Eq.~\eqref{eq:sigma}. They can then use the hashing protocol on these states to distill entanglement with rate $I_{\to}^c(\sigma_{A'B})/d^2_A$. Finally, since $\ket{\psi_{AA'}}$ can be any pure state, we get the required bound:
\begin{equation}
    D_{\to}(J_{AB}(\Phi)) \geq \frac{1}{d_A^2} \max_{\ket{\psi_{AA'}}} I^c_{\to} (\sigma_{A'B}) = \frac{1}{d_A^2} \mathcal{Q}^{(1)}(\Phi). \quad \qedsymbol
\end{equation}

\section{Proof of Theorem~\ref{theorem:main1}} \label{appen:th-main1}
Assume that Alice and Bob initially share $N$ copies of $\rho_{AB}$. On each such copy, Bob performs a generalized measurement with measurement operators $\{Y_B, \sqrt{\iden - Y_B^\dagger Y_B}\}$, where $Y_B=\sqrt{\lambda^B_{\min}} \rho_B^{-1/2}$ and the measurement succeeds if the outcome corresponding to $Y_B$ is observed. This happens with probability $p_{succ} = \Tr [(\iden_A\otimes Y_B) \rho_{AB} (\iden_A\otimes Y_B^\dagger)] = \lambda^B_{\min} r_B$, in which case the post-measurement state is given by
\begin{equation}
    \rho'_{AB} = \frac{1}{p_{succ}} (\iden_A\otimes Y_B) \rho_{AB} (\iden_A\otimes Y_B^\dagger).
\end{equation}
Bob conveys the measurement results to Alice via classical communication, so that they can discard the states for which the measurement failed. After this step, Alice and Bob share $\sim p_{succ}N$ copies of $\rho'_{AB}$. Note that $\rho'_B = \Pi_B/r_B$, where $\Pi_B$ is the orthogonal projection onto the support of $\rho_B$. Alice and Bob can now use the hashing protocol on the filtered states to distill entanglement with rate
\begin{align}
    p_{succ}I^c_{\to}(\rho'_{AB})  &= \lambda^B_{\min} r_B
    [S(\rho'_B) - S(\rho'_{AB})] \nonumber \\
    &\geq \lambda^B_{\min} r_B [\log r_B - \log r]. \quad \qedsymbol
\end{align}

\section{Proof of Eq.~\eqref{eq:rankbound}}\label{appen:eq-rankbound}
Clearly, the map 
\begin{align}
    \Lambda: \M{d_A} &\to \M{d_B} \nonumber \\
    \sigma_A &\mapsto \operatorname{Tr}_A [(\sigma_A\otimes \iden_B)\rho_{AB}]
\end{align}
is positive. For any (normalized) $\ket{\phi_A}$, since $\ketbra{\phi_A}\leq \iden_A$, we get $\rho^\phi_B = \Lambda(\ketbra{\phi_A}) \leq  \Lambda(\iden_A) = \rho_B$. Hence, $\rank\rho^\phi_B \leq \rank \rho_B$. Now, if $\sigma_{AC}$ is complementary to $\rho_{AB}$, we can apply the previous argument on $\sigma_{AC}$ to deduce that
$\rank\rho^\phi_B = \rank\sigma^\phi_C \leq \rank \sigma_C = \rank \rho_{AB}. \quad \qedsymbol$

\section{Proof of Theorem~\ref{theorem:main2}} \label{appen:th-main2}

Let us first state the result from \cite{Siddhu2021logsingularity, Singh2022detecting} which will be crucially used in the following proof. Note that for a quantum channel $\Phi:\M{d_A}\to \M{d_B}$, the following analogue of Eq.~\eqref{eq:rankbound} holds for all pure input states $\ketbra{\phi_A}$: 
\begin{equation*}
    \operatorname{rank}\Phi(\ketbra{\phi_A})\leq \min \{ \operatorname{rank}J_{AB}(\Phi), \operatorname{rank}J_B (\Phi) \}.
\end{equation*}

\begin{lemma} \cite[Theorem 1]{Siddhu2021logsingularity} \cite[Corollary II.8]{Singh2022detecting} \label{lemma:D}
    Let $\Phi:\M{d_A}\rightarrow \M{d_B}$ be a channel such that $\operatorname{rank}J_{AB}(\Phi)< \operatorname{rank}J_B (\Phi)$ and there exists a pure input state $\ketbra{\phi_A}$ with $\operatorname{rank}\Phi(\ketbra{\phi_A})= \min \{ \operatorname{rank}J_{AB}(\Phi), \operatorname{rank}J_B (\Phi) \} = \operatorname{rank}J_{AB}(\Phi)$. Then, $\mathcal{Q}^{(1)}(\Phi)>0$.
\end{lemma}

We are now ready to prove Theorem~\ref{theorem:main2}. Assume that Alice and Bob initially share $N$ copies of $\rho_{AB}$. On each such copy, Alice performs a generalized measurement with measurement operators $\{Y_A, \sqrt{\iden - Y_A^\dagger Y_A}\}$, where $Y_A=\sqrt{\lambda^A_{\min}} \rho_A^{-1/2}$ and $\lambda^A_{min}$ is the minimum positive eigenvalue of $\rho_A$. The measurement succeeds if the outcome corresponding to $Y_A$ is observed (this happens with probability $p_{succ} = \Tr [(Y_A\otimes \iden_B) \rho_{AB} (Y^\dagger_A\otimes \iden_B)] > 0 $), in which case the post-measurement state is
\begin{equation}
    \rho'_{AB} = \frac{1}{p_{succ}} (Y_A \otimes \iden_B) \rho_{AB} (Y_A^\dagger \otimes \iden_B).
\end{equation}
Note that $\rho'_A = \Pi_A/r_A$, where $\Pi_A$ projects orthogonally onto the support of $\rho_A$ and $r_A=\operatorname{rank}\rho_A$. By using the teleportation trick explained in the proof of Eq.~\eqref{eq:cap-distill-bound2}, Alice and Bob can use $\rho'_{AB}$ to probabilistically implement the channel $\Phi_{\rho'}: \M{r_A}\to \M{d_B}$ defined by $J_{AB}(\Phi_{\rho'})=\rho'_{AB}$. Moreover, $\rank J_{AB}(\Phi_{\rho'})= r < r_B = \rank J_B(\Phi_{\rho'})$ (recall that if $A,B\in \M{d}$ are positive semi-definite with $\operatorname{supp}B\subseteq \operatorname{supp}A$, then $\rank ABA^{\dagger} = \operatorname{rank}B$), and 
\begin{equation}
    \Phi_{\rho'}(X) = \Phi_{\rho} (Y_A^\top X Y_A^\top),
\end{equation}
where $\Phi_\rho$ is a completely positive (not necessarily trace-preserving) map defined by $J_{AB}(\Phi_{\rho})=\rho_{AB}$. Hence, the constraint $\operatorname{rank}\rho^\phi_B = r$ translates into the following:
\begin{align*}
    \rank\Phi_{\rho'} (Z_A \ketbra{\phi_A} Z_A)&= \rank \Phi_{\rho}(\ketbra{\phi_A}) \\ 
    & = \rank \rho^{\phi}_B = \min\{r, r_B\}=r,
\end{align*}
where $Z_A =(Y^\top_A)^{-1}.$
Lemma~\ref{lemma:D} then shows that $\Phi_{\rho'}$ can transmit quantum information at a non-zero rate: $\mathcal{Q}^{(1)}(\Phi_{\rho'})>0$. Hence, we have
\begin{equation}
    D_{\to}(\rho_{AB})\geq p_{succ}D_{\to}(\rho'_{AB})\geq \frac{p_{succ}}{r_A^2}\mathcal{Q}^{(1)}(\Phi_{\rho'}) > 0,
\end{equation}
where the penultimate inequality follows from Eq.~\eqref{eq:cap-distill-bound2}. $\qedsymbol$

\section{Proof of Theorem~\ref{theorem:main3}}\label{appen:th-main3}
It is easy to prove that for $(\M{d_A}\otimes \M{d_B})\ni \rho_{AB}\sim \mu_{d_C}$, the following statements are almost surely true:
\begin{align}
    \operatorname{rank}\rho_{AB} = \min\{d_C, d_A d_B \}, \nonumber \\
    \operatorname{rank} \rho_B = \min \{d_B, d_A d_C \}.
\end{align}
Hence, for $\rho_{AB}\sim \mu_{d_C}$ and $d_C<d_B$, $\operatorname{rank}\rho_{AB}<\operatorname{rank}\rho_B$ almost surely. Moreover, a random pure state $\ket{\psi_{ABC}}$ is nothing but a (suitably normalized) $d_B d_C \times d_A$ random matrix with i.i.d. standard complex Gaussian entries (see \cite{Nechita2007random, Zyczkowski2011random}). Hence, each of its columns $\ket{V^i_{BC}}$ for $i=1,2,\ldots ,d_A$, considered as vectors in $\C{d_B}\otimes \C{d_C}$, have full Schmidt rank =$\min \{d_B, d_C \}=d_C$ almost surely. In other words, for any basis vector $\ket{i_A}\in \C{d_A}$, we have 
\begin{equation}
    \rank \rho^i_B = \rank \operatorname{Tr}_C\ketbra{V^i_{BC}} = d_C
\end{equation}
almost surely (recall that $\rho^i_B$ was defined below Eq.~\eqref{eq:rankbound}). A simple application of Theorem~\ref{theorem:main2} then gives us the desired result. $\quad \qedsymbol$


\end{document}